\documentclass[11pt]{article}

\usepackage{graphicx}
\usepackage{amsfonts}
\usepackage{amsmath}
\usepackage{amsthm}
\usepackage[mathcal]{euscript}      
\usepackage{bm}                     
\usepackage{color}
\usepackage{epstopdf}
\usepackage{amssymb}
\usepackage{mathrsfs,mathdots}
\usepackage{lscape}
\usepackage[version=3]{mhchem}
\usepackage{multicol}

\usepackage[paper=a4paper,dvips,top=3cm,left=2.4cm,right=2.4cm,
    foot=1cm,bottom=3cm]{geometry}

\begin{document}

\title{Irreducible Function Bases of Isotropic Invariants of a Third Order Three-Dimensional Symmetric and Traceless Tensor}
\author{Yannan Chen\footnote{%
    School of Mathematics and Statistics, Zhengzhou University, Zhengzhou 450001, China ({\tt ynchen@zzu.edu.cn}).
    This author was supported by the National Natural Science Foundation of China (Grant No. 11571178, 11771405)
    and the Hong Kong Polytechnic University Postdoctoral Fellowship.}
\and Shenglong Hu\footnote{Department of Mathematics, School of Science, Hangzhou Dianzi University, Hangzhou 310018, China. ({\tt shenglonghu@hdu.edu.cn}). This author was supported by National Science Foundation of China (Grant No. 11771328).}
\and Liqun Qi\footnote{%
    Department of Applied Mathematics, The Hong Kong Polytechnic University,
    Hung Hom, Kowloon, Hong Kong ({\tt maqilq@polyu.edu.hk}).
    This author was supported by the Hong Kong Research Grant Council
    (Grant No. PolyU  15302114, 15300715, 15301716 and 15300717).}
\and Wennan Zou\footnote{%
    Institute for Advanced Study, Nanchang University, Nanchang 330031, China ({\tt zouwn@ncu.edu.cn}).
    This author was supported by the National Natural Science Foundation of China (Grant No. 11372124)}
    }

\date{\today}
\maketitle

\begin{abstract}  Third order three-dimensional symmetric and traceless tensors play an important role in physics and tensor representation theory.   A minimal integrity basis of a third order three-dimensional symmetric and traceless tensor has four invariants with degrees two, four, six and ten respectively.   In this paper, we show that any minimal integrity basis of a third order three-dimensional symmetric and traceless tensor is also an irreducible function basis of that tensor, and there is no polynomial syzygy relation among the four invariants of that basis, i.e., these four invariants are algebraically independent.

\end{abstract}

\newtheorem{Theorem}{Theorem}[section]
\newtheorem{Definition}[Theorem]{Definition}
\newtheorem{Lemma}[Theorem]{Lemma}
\newtheorem{Corollary}[Theorem]{Corollary}
\newtheorem{Proposition}[Theorem]{Proposition}
\newtheorem{Conjecture}[Theorem]{Conjecture}
\newtheorem{Question}[Theorem]{Question}

\renewcommand{\hat}[1]{\widehat{#1}}
\renewcommand{\tilde}[1]{\widetilde{#1}}
\renewcommand{\bar}[1]{\overline{#1}}
\newcommand{\REAL}{\mathbb{R}}
\newcommand{\COMPLEX}{\mathbb{C}}
\newcommand{\SPHERE}{\mathbb{S}^2}
\newcommand{\diff}{\,\mathrm{d}}
\newcommand{\st}{\mathrm{s.t.}}
\newcommand{\T}{\top}
\newcommand{\vt}[1]{{\bf #1}}
\newcommand{\aaa}{{\vt{a}}}
\newcommand{\ddd}{{\vt{d}}}
\newcommand{\x}{{\vt{x}}}
\newcommand{\y}{{\vt{y}}}
\newcommand{\z}{{\vt{z}}}
\newcommand{\uu}{{\vt{u}}}
\newcommand{\vv}{{\vt{v}}}
\newcommand{\ww}{{\vt{w}}}
\newcommand{\e}{{\vt{e}}}
\newcommand{\g}{{\vt{g}}}
\newcommand{\0}{{\vt{0}}}
\newcommand{\Ten}{\bf{T}}
\newcommand{\HH}{\mathbb{H}}
\newcommand{\A}{{\bf A}}
\newcommand{\B}{{\bf B}}
\newcommand{\C}{\mathcal{C}}
\newcommand{\D}{{\bf D}}
\newcommand{\E}{\bf{E}}
\newcommand{\OOO}{\mathcal{O}}
\newcommand{\U}{{\bf{U}}}
\newcommand{\V}{{\bf{V}}}
\newcommand{\W}{{\bf{W}}}
\newcommand{\I}{{\bf{I}}}
\newcommand{\II}{{\mathcal{I}}}
\newcommand{\OO}{{\bf{O}}}
\newcommand{\RESULTANT}{\mathrm{Res}}

\newcommand{\supercite}[1]{ \cite{#1}}

\section*{\small Nomenclature}
\footnotesize \setlength{\columnsep}{5mm}
\begin{description}
\item[$\D$] a third order three-dimensional symmetric and traceless tensor with components $D_{ijk}$
\item[$\operatorname{T}(m,n)$] the space of real tensors of order $m$ and dimension $n$
\item[$\operatorname{S}(m,n)$] the subspace of symmetric tensors
\item[$\operatorname{St}(m,n)$] the subspace of symmetric and traceless tensors
\item[$\operatorname{O}(n)$] the orthogonal group of dimension $n$
\item[$\operatorname{SO}(n)$] the special orthogonal group of dimension $n$
\item[$\operatorname{Gl}(n,\mathbb R)$] the general linear group of real matrices
\item[${m \choose n} = \frac{m!}{n!(m-n)!}$] the binomial coefficient for $m\ge n\ge 0$
\end{description}

\newpage
\section{Introduction}
\small
Third order three-dimensional symmetric and traceless tensors play an important role in physics and tensor representation theory.   In the study of liquid crystal, they are used to characterize condensed phases exhibited by bent-core molecules\supercite{CQV-18, GV-16, QCC-18}.   In tensor representation theory, a tensor space is called $\operatorname{O}(3)$-stable if any orthogonal transformation converts that space to itself.  The space of symmetric and traceless tensors of some order is $\operatorname{O}(3)$-stable and does not contain any proper $\operatorname{O}(3)$-stable subspace. Hence, the space of third order three-dimensional symmetric and traceless tensors is a fundamental tensor space.

In 1997, Smith and Bao\supercite{SB-97} presented a minimal integrity basis of a third order symmetric and traceless tensor.   The Smith-Bao minimal integrity basis has four invariants with degrees two, four, six and ten respectively.    It is known that the number of invariants with the same degree in a minimal integrity basis of some tensors is always fixed\supercite{OKA-17}.
Thus, any  minimal integrity basis of a third order symmetric and traceless tensor has four invariants with degrees two, four, six and ten respectively.

In this paper, we show that any minimal integrity basis of a third order three-dimensional symmetric and traceless tensor is also an irreducible function basis of that tensor, and there is no polynomial syzygy relation among the four invariants of that basis, i.e., these four invariants are algebraically independent\supercite{Sh-67}.

In the next section, some preliminaries are given.

In Section 3, we give a proof for the cardinality of a function basis of the invariants for a finite dimensional real vector space by a compact group is bounded below by the intuitive difference of the dimensions of the vector space and the group. Applying this result to the space of third order three-dimensional symmetric and traceless tensors, we show that each minimal integrity basis of a third order three-dimensional symmetric and traceless tensor is also an irreducible function basis of that tensor.

Then, in Section 4, we further show that there is no polynomial syzygy relation among the four invariants of any minimal integrity basis of a third order three-dimensional symmetric and traceless tensor.   In the other words, these four invariants are algebraically independent\supercite{Sh-67}.

The results of this paper enrich the knowledge about minimal integrity bases and irreducible function bases of third order three-dimensional tensors.  In the last section, we review the progresses in this area.


\section{Preliminaries}\label{sec:preliminary}
In this section, we present necessary notions and results from tensor invariant theory and summarize the results about minimal integrity bases of a third order three-dimensional symmetric and traceless tensor.
\subsection{Tensor Invariants}\label{sec:tensor}
Let $m>1$ and $n>1$ be given integers. The space of real tensors $\mathcal A$ of order $m$ and dimension $n$ is formed by all tensors (a.k.a. hypermatrices) with entries $a_{i_1\dots i_m}\in\mathbb R$, the field of real numbers, for all $i_j\in\{1,\dots,n\}$ and $j\in\{1,\dots,m\}$. It is denoted as $\operatorname{T}(m,n)$. Let $\operatorname{Gl}(n,\mathbb R)\subset\mathbb R^{n\times n}$ be the \textit{general linear group} of real matrices. Let $G\subseteq \operatorname{Gl}(n,\mathbb R)$ be a subgroup. We then have a natural group representation $G\rightarrow \operatorname{Gl}(\operatorname{T}(m,n),\mathbb R)$, the real general linear group of the linear space $\operatorname{T}(m,n)$, via
\[
(g\cdot\mathcal T)_{j_1\dots j_m}:=\sum_{i_1}^n\dots\sum_{i_m=1}^ng_{j_1i_1}\dots g_{j_mi_m}t_{i_1\dots i_m}.
\]
A linear subspace $V$ of $\operatorname{T}(m,n)$ is \textit{$G$-stable} if $g\cdot v\in V$ for all $g\in G$ and $v\in V$.

Of particular interests in this article are the compact subgroups $\operatorname{O}(n)$ (the orthogonal group) and $\operatorname{SO}(n)$ (the special orthogonal group), both of which are Lie groups\supercite{Ha-03}.

In $\operatorname{T}(m,n)$, the subspace of symmetric tensors $\operatorname{S}(m,n)$ is $\operatorname{Gl}(n,\mathbb R)$-stable, and thus $G$-stable for every subgroup $G$. Likewise, inside $\operatorname{S}(m,n)$, the subspace of symmetric and traceless tensors $\operatorname{St}(m,n)$ is $\operatorname{O}(n)$-stable, thus $\operatorname{SO}(n)$-stable. Recall that a symmetric tensor $\mathcal T\in \operatorname{S}(m,n)$ is traceless if
\[
\sum_{i=1}^nt_{iii_3\dots i_m}=0\ \text{for all }i_3,\dots,i_m\in\{1,\dots,n\}.
\]
A well-known fact is that the dimension of $\operatorname{S}(m,n)$ as a linear space is ${n+m-1\choose n-1}$, and that of $\operatorname{St}(m,n)$ is ${n+m-1\choose n-1}-{n+m-3\choose n-1}$.

Associated to a linear subspace $V\subseteq\operatorname{T}(m,n)$ is an algebra $\mathbb R[V]$, generated by the dual basis of $V$. Once a basis of $V$ is fixed, an element $f\in \mathbb R[V]$ can be viewed as a polynomial in terms of the coefficients of $v\in V$ in that basis. Let $G\subseteq \operatorname{Gl}(n,\mathbb R)$ be a subgroup and $V$ be $G$-stable. Then, we can induce a group action of $G$ on $\mathbb R[V]$ via
\[
(g\cdot f)(v)=f(g^{-1}\cdot v)\ \text{for all }g\in G\ \text{and }v\in V.
\]
With this group action, some elements of $\mathbb R[V]$ are fixed points for the whole $G$, i.e.,
\[
g\cdot f=f\ \text{for all }g\in G,
\]
which form a subring $\mathbb R[V]^G$ of $\mathbb R[V]$\supercite{Ol-99,We-39}. Elements of $\mathbb R[V]^G$ are \textit{invaraints} of $V$ under the action of $G$. It is well-known that $\mathbb R[V]^G$ is finitely generated. A generator set is called an \textit{integrity basis}.  In an integrity basis, if none of the generators is a polynomial of the others, it is a \textit{minimal integrity basis}. Given a subspace $V$ and group $G$, minimal integrity bases may not be unique, but their cardinalities are the same as well as the lists of degrees of the generators \supercite{Spe-71}. Invariants in $\mathbb R[V]^G$ are polynomials, always referred as \textit{algebraic invaraints}.

Likewise, one can consider \textit{function invariants}\supercite{Ol-99}. A function $f\colon V\rightarrow\mathbb R$ is an invariant if
\[
f(v)=f(g\cdot v)\ \text{for all }g\in G.
\]
The set of function invariants of $V$ is denoted as $\mathcal I(V)$.
If there is a set of generators such that each function invariant can be expressed as a function of the generators, it is called a \textit{function basis}. Similarly, if none of the generators is a function of the others in a function basis, it is called an \textit{irreducible function basis}.

\subsection{Minimal Integrity Bases of a Third Order Three-Dimensional Symmetric and Traceless Tensor}\label{sec:integrity}

Use $\D$ to denote a third order three-dimensional symmetric and traceless tensor. From now on, the summation convention is used, i.e., in a product, if an index is repeated twice, then it is summed up from $1$ to $3$ for that index.

In 1997, Smith and Bao\supercite{SB-97} presented a minimal integrity basis for $\D$ as $\{I_2, I_4, I_6, I_{10} \}$, with
\begin{equation*}
\begin{array}{lll}
  I_2 := D_{ijk}D_{ijk}, & I_4 := D_{ijk}D_{ij\ell}D_{pqk}D_{pq\ell}, \\
  I_6 := v_i^2, & I_{10} := D_{ijk}v_iv_jv_k,
\end{array}
\end{equation*}
where $v_p:=D_{ijk}D_{ij\ell}D_{k\ell p}$.

The number of invariants with the same degree in a minimal integrity basis of some tensors is always fixed\supercite{OKA-17}.
Hence, any  minimal integrity basis of $\D$ has four invariants with degrees two, four, six and ten respectively.
We denote the four invariants of a general  minimal integrity basis of $\D$ by $J_2, J_4, J_6$ and $J_{10}$, respectively.

\section{Irreducible Function Bases of A Third Order Symmetric and Traceless Tensor}

In this section, we show that the cardinality of a function basis of the invariants for a finite dimensional real vector space by a compact group is lower bounded by the intuitive difference of the dimensions of the vector space and the group.  Then we apply this result to the space of third order three dimensional symmetric and traceless tensors, showing that each minimal integrity basis of a third order three-dimensional symmetric and traceless tensor is also an irreducible function basis of that tensor.

\subsection{Quotient Manifold by Lie Groups}\label{sec:quotient}
A real vector space $V$ of finite dimension has a natural manifold structure. Any given equivalence relation $\sim$ on $V$ defines a quotient structure with elements being the \textit{equivalence classes}
\[
V/\sim:=\{[v]\mid v\in V\}\ \text{with }[v]:=\{u\in V\mid v\sim u\}.
\]
The set $V/\sim$ is the \textit{quotient} of $V$ by $\sim$, and $V$ is the \textit{total space} of $V/\sim$.
The quotient $V/\sim$ is a \textit{quotient manifold} if the natural projection $\pi : V\rightarrow V/\sim$ is a submersion.
$V/\sim$ admits at most one manifold structure making it being a quotient manifold$^{\text{\cite[Proposition~3.4.1]{AMS-08}}}$. It may happen that $V/\sim$ has a manifold structure but fails to be a quotient manifold. Whenever $V/\sim$ is indeed a quotient manifold, we call the equivalence relation $\sim$ \textit{regular}.

Let $G$ be any compact Lie group and $V$ a finite dimensional real linear space. Suppose that $V$ is a representation of $G$, i.e., there is a group homomorphism $G\rightarrow\operatorname{Gl}(V,\mathbb R)$. Then, there is a natural equivalence relation given by $G$ as
\[
v\sim u\ \text{if and only if }g\cdot v=u\ \text{for some }g\in G. 
\]
The quotient under this equivalence is sometimes denoted as $V/G$, which is the set of orbits of the group action of $G$ on $V$. Suppose in the following that the group action is continuous.
Then, with the compactness of $G$, it can be shown that $V/G$ is a quotient smooth manifold, since the graph set
\[
\{(v,u)\mid [v]=[u]\}\subset V\times V
\]
is closed$^{\text{\cite[Proposition~3.4.2]{AMS-08}}}$.

Note that the fibre of the natural projection $\pi$ is the equivalence class $\pi^{-1}(\pi(v))=[v]$ for each $v\in V$. If $[v]$ is not a discrete set of points for some $v\in V$, then the dimension of $V/\sim$ is strictly smaller than the dimension of $V$$^{\text{\cite[Proposition~3.4.4]{AMS-08}}}$.

In the following, we consider subspaces of  the linear space of tensors of order $m$ and dimension $n$, i.e., $V\subseteq \operatorname{T}(m,n)$.
\begin{Lemma}\label{lem:dimension}
Let $V\subseteq\operatorname{T}(m,n)$ be a linear space containing $\operatorname{St}(m,n)$ and $G=\operatorname{O}(n)$ or $\operatorname{SO(n)}$. Then, we have $\operatorname{dim}(V/G)<\operatorname{dim}(V)$, and
\begin{equation}\label{eq:dim}
\operatorname{dim}(V/G)\geq \operatorname{dim}{V}-\operatorname{dim}(G).
\end{equation}
\end{Lemma}
\begin{proof}
By Proposition~3.4.4 in book \supercite{AMS-08}, if there is one point $v\in V$ such that $[v]$ is not a set of discrete points, then $\operatorname{dim}(V/G)<\operatorname{dim}(V)$, and $\operatorname{dim}(V/G)=\operatorname{dim}{V}-\operatorname{dim}([v])$, where $[v]$ is regarded as an embedded submanifold of $V$.

Note that $[v]$ is the orbit of $G$ acting on the element $v$. Thus, the dimension of $[v]$ cannot exceed the dimension of $G$. Consequently, the dimension bound \eqref{eq:dim} follows if we can find a point $v\in V$ such that $[v]$ is not a discrete set of points.

First of all, we show that $[v]$ cannot be a discrete set of points for the group $G=\operatorname{SO(n)}$ for some $v\in V$.

It is easy to see that the stabilizers $G_v=G$ cannot hold through out $v\in V$. Thus, there exists an orbit $[v]$ with more than one element. Suppose that $[v]$ is a discrete set of more than two points.
For any given two discrete points $v_1,v_2\in [v]$, there exist $g_1,g_2\in G$ such that
\[
v_i=g_i\cdot v\ \text{for all }i=1,2
\]
by the definition of $[v]$. Since $\operatorname{SO(n)}$ is a connected manifold, there is a smooth curve $g(t)$ starting from $g(0)=g_1$ ending at $g(1)=g_2$. By the definition,
\[
g(t)\cdot v\in [v]\ \text{for all }t\in [0,1].
\]
Since the group action is smooth, we see that $v_1$ and $v_2$ is thus connected, contradicting the discreteness.

Since $\operatorname{SO(n)}$ is one half connected component of  $\operatorname{O(n)}$, the result for  $\operatorname{O(n)}$ follows immediately.
\end{proof}


\subsection{Cardinality of Function Basis}\label{sec:function}

The next result is Theorem~11.112 in book\supercite{Vi-03}, see also the classical book \supercite{We-39}.
\begin{Lemma}[Separability]\label{lem:sep}
Let $G$ be a compact group and $V$ a real vector space representing $G$. Then the orbits of $G$ acting on $V$ are separated by the invariants $\mathbb R\mathbb [V]^G$.
\end{Lemma}
The conclusion may fail in the complex case.

The concepts of function invariants and functional independence of invariants can be found in classical textbooks, see for example$^{\text{\cite[Page 73]{Ol-99}}}$.

The analysis for integrity and minimal integrity bases of $V$ for some $G$ is more sophisticated and approachable than function basis. Nevertheless, an exciting fact that an integrity basis is also a function basis holds in most interesting cases.
We will present this result in Theorem~\ref{thm:function}.
\begin{Theorem}[Function Basis]\label{thm:function}
Let $G$ be a compact group and $V$ a finite dimensional real linear vector space representing $G$. Then, any integrity basis of $\mathbb R[V]^G$ is a function basis.
\end{Theorem}

\begin{proof}It is well-known that the ring of polynomial invariants $\mathbb R[V]^{G}$ is finitely generated, whose minimal set of generators is an integrity basis.

The orbits of $G$ on $V$ are separable, i.e., $p(u)=p(v)$ for all $p\in \mathbb R[V]^{G}$ if and only if $u=g\cdot v$ for some $g\in G$ by Lemma~\ref{lem:sep}.  Let $\mathcal P:=\{p_1,\dots,p_r\}$ be an integrity basis. We have a map
\[
\mathbb P: V\rightarrow \mathbb P(V)\ \text{with }v\mapsto (p_1(v),\dots,p_r(v))^\mathsf{T},
\]
where $\mathbb P(V)$ is the image of $\mathbb P$ on $V$.
Actually, this map is defined over $V/G$, as each $p_i\in \mathcal P$ is an invariant. Moreover, this map, with $V/G\rightarrow\mathbb P(V)$, is onto and one to one, following from the separability of $\mathbb R[V]^{G}$ on $V$ and the fact that each algebraic invariant is generated by $p_1,\dots,p_r$. Thus, there is an inverse map
\[
\mathbb P^{-1}: \mathbb P(V)\rightarrow V/G.
\]
In summary,
we can conclude that $[v]$ (the equivalent class in $V/G$) for any $v\in V$ can be determined by the values of $p_1(v),\dots,p_r(v)$.
On the other side, each invariant in $\mathcal I(V)$, the set of invariants of $V$, is a function over $V/G$. Thus, we have a chain of functions
\[
V\rightarrow\mathbb P(V)\leftrightarrow V/G\rightarrow \mathbb R.
\]
Reading throughout the above chain, we get that the integrity basis $\mathcal P$ gives a function basis for $\mathcal I(V)$.
\end{proof}

When conditions in Theorem~\ref{thm:function} are fulfilled, we can derive a function basis and even an irreducible function basis from an integrity basis or minimal integrity basis. A function basis derived from an integrity basis is called a \textit{polynomial function basis}, and an irreducible function basis derived from a minimal integrity basis is called an \textit{irreducible polynomial function basis}. Note that any function basis consisting of polynomial invariants is a polynomial function basis as it can always be expanded to an integrity basis.  In the following, we will give a lower bound for the cardinality of a polynomial function basis.

Since $\mathbb R[V]^G$ is finitely generated\supercite{Ha-03} and has no nilpotent elements, it follows from$^{\text{\cite[Theorem~1.3]{Sh-77}}}$ that
that $V/G$ is a (quotient) variety. It is the variety determined by the coordinate ring $\mathbb R[V]/(\mathbb R[V]^G)$.

\begin{Theorem}[The Cardinality Theorem]\label{thm:card}
Let $G$ be a compact group of dimension $d$ and $V$ a finite dimensional real linear vector space representing $G$ of dimension $N>d$. Then, any polynomial function basis has cardinality being not smaller than $N-d$.
\end{Theorem}

\begin{proof}
Let $\{p_1,\dots,p_r\}\subset P[V]^{G}$ be a polynomial function basis. We must have that for each pair $u,v\in V$
\[
p_i(u)=p_i(v)\ \text{for all }i\in\{1,\dots,r\}
\]
will implies
\[
[u]=[v],
\]
since each polynomial in $P[V]^{G}$ is a function of $p_1,\dots,p_r$, and $P[V]^G$ separates the orbits of $V/G$ \supercite{We-39}.

We therefore have that the mapping
\[
\mathcal P : V/G\rightarrow \mathbb R^r
\]
given by
\[
\mathcal P([v])=(p_1(v),\dots,p_r(v))^\mathsf{T}
\]
is a one to one regular map. Obviously, we can consider the mapping
\[
\mathcal P : V/G\rightarrow \overline{\mathcal P(V/G)}\subseteq\mathbb R^r
\]
whenever $\mathcal P$ is not dominant. Now, the map
\[
\mathcal P : V/G\rightarrow \overline{\mathcal P(V/G)}
\]
is a dominant morphism. Then, if $r<N-d\leq \operatorname{dim}(V/G)$, each fibre of $\mathcal P^{-1}(\mathbf y)$ for $\mathbf y\in \mathcal P(V/G)$ will have dimension at least $\operatorname{dim}(V/G)-\operatorname{dim}(\overline{\mathcal P(V/G)})\geq N-d-r\geq 1$ $^{\text{\cite[Proposition~6.3]{b}}}$. This contradicts the separability of the set $\{p_1,\dots,p_r\}$ on the orbits of $V/G$ immediately.
\end{proof}

\subsection{Irreducible Function Bases of A Third Order Symmetric and Traceless Tensor}\label{sec:app}
By the cardinality theorem for function basis, we have the following result for third order three-dimensional symmetric and traceless tensors.

\begin{Theorem}\label{prop:33}
Every minimal integrity basis of isotropic invariants of a third order three-dimensional symmetric and traceless tensor $\D$ is an irreducible function basis of that tensor.
\end{Theorem}

\begin{proof}
First note that the dimension of $\operatorname{St}(3,3)$ is $7$. Thus, the dimension of $\operatorname{St}(3,3)/\operatorname{O}(3)$ is at least $4$.  It follows from Theorem~\ref{thm:card} that an irreducible function basis will have cardinality at least $4$.

On the other hand, every minimal integrity basis of $\operatorname{St}(3,3)$ will have the same cardinality $4$ \supercite{Spe-71}, which is of course an upper bound for the cardinality of irreducible function bases derived from them.

As the lower bound is equal to the upper bound for the cardinality of the irreducible function basis, the conclusion follows.
\end{proof}

\textbf{Remark.} We may directly show that the Smith-Bao minimal integrity basis $\{ I_2, I_4, I_6, I_{10}\}$ is an irreducible function basis of a third order three-dimensional symmetric and traceless tensor $\D$ by using the method proposed in\supercite{PT-87}. Since a minimal integrity basis is also a function basis, we only need to prove that none of $\{ I_2, I_4, I_6, I_{10} \}$ is a single-valued function of the others.

Using seven independent elements of the tensor $\D$:
\begin{equation*}
  D_{111}, D_{112}, D_{113}, D_{122}, D_{123}, D_{222}, \text{ and } D_{223},
\end{equation*}
we represent the multi-way array corresponding to $\D$ as
\begin{equation*}\footnotesize
  \left(\begin{array}{ccc|ccc|ccc}
   D_{111} & D_{112} & D_{113}          & D_{112} & D_{122} & D_{123}          & D_{113} & D_{123} & -D_{111}-D_{122} \\
    D_{112} & D_{122} & D_{123}          & D_{122} & D_{222} & D_{223}          & D_{123} & D_{223} & -D_{112}-D_{222} \\
    D_{113} & D_{123} & -D_{111}-D_{122} & D_{123} & D_{223} & -D_{112}-D_{222} & -D_{111}-D_{122} & -D_{112}-D_{222} & -D_{113}-D_{223} \\
  \end{array}\right).
\end{equation*}

Let $D_{111} = \sqrt[4]{3}$, $D_{112} = D_{113} = D_{122} = D_{123} = D_{222} = D_{223} = 0$.  Then $I_2 = 4\sqrt{3}, I_4 = 24, I_6 = I_{10} = 0$.
Let $D_{112} =  \sqrt[4]{2}$, $D_{111} = D_{113} = D_{122} = D_{123} = D_{222} = D_{223} = 0$.  Then $I_2 = 6\sqrt{2}, I_4 = 24, I_6 = I_{10} = 0$.  We see that with respect to these two examples, the values of $I_4, I_6$ and $I_{10}$ keep invariant, but the value of $I_2$ is changed.   This shows that $I_2$ is not a function of $I_4, I_6$ and $I_{10}$.

Let $D_{111} = \sqrt{3}$, $D_{112} = D_{113} = D_{122} = D_{123} = D_{222} = D_{223} = 0$.  Then $I_2 = 12, I_4 = 72, I_6 = I_{10} = 0$.
Let $D_{112} =  \sqrt{2}$, $D_{111} = D_{113} = D_{122} = D_{123} = D_{222} = D_{223} = 0$.  Then $I_2 = 12, I_4 = 48, I_6 = I_{10} = 0$.  We see that with respect to these two examples, the values of $I_2, I_6$ and $I_{10}$ keep invariant, but the value of $I_4$ is changed.   This shows that $I_4$ is not a function of $I_2, I_6$ and $I_{10}$.

Let $D_{111} = D_{112} = 1$, $D_{113} = D_{122} = D_{123} = D_{222} = D_{223} = 0$.  Then $I_2 = 10, I_4 = 44, I_6 = 16, I_{10} = 64$.
Let $D_{111} = D_{123} = 1$, $D_{112} = D_{113} = D_{122} = D_{222} = D_{223} = 0$.  Then $I_2 = 10, I_4 = 44, I_6 = 16, I_{10} = -64$.  We see that with respect to these two examples, the values of $I_2, I_4$ and $I_6$ keep invariant, but $I_{10}$ changes its sign.   This shows that $I_{10}$ is not a function of $I_2, I_4$ and $I_6$.

Let $f(t)=-43+\cos(6t)+84\sin(3t)$. Since $f(0)f(\frac{\pi}{6})=-42\cdot40<0$, we know $f(t)=0$ has a root in $(0,\frac{\pi}{6})$, which is denoted as $t_0$.
Let $D_{111}=1$, $D_{122}=-\frac{1}{2}+\frac{1}{2}\sin(t_0)$, $D_{123}=\frac{1}{2}\cos(t_0)$, $D_{223}=-2$, $D_{112} = D_{113} = D_{222} = 0$.
Then, $I_2 = 20, I_4 = 176, I_6=104-24\sin(3t_0), I_{10}=-16(-43+\cos(6t_0)+84\sin(3t_0))=0.$
On the other hand, let $D_{111} = D_{112} = D_{113} = D_{123} = 1$, $D_{122} = D_{222} = D_{223} = 0$. Then $I_2 = 20, I_4 = 176, I_6 = 128, I_{10} = 0$.
Clearly, since $t_0\in (0,\frac{\pi}{6})$, we have
$$104-24\sin(3t_0)< 104 <128.$$ Hence,  $I_6$  is not a function of $I_2, I_4$ and $I_{10}$.

Hence, none of $I_2, I_4, I_6$ and $I_{10}$ is a function of the other three invariants, i.e., $\{ I_2, I_4, I_6, I_{10} \}$ is also an irreducible function basis of a third order three-dimensional symmetric and traceless tensor $\D$.

Theorem \ref{prop:33} claims that any minimal integrity basis  of a third order three-dimensional symmetric and traceless tensor $\D$ is an irreducible function basis of that tensor.  Hence, Theorem \ref{prop:33} is more general.   The above direct proof for the Smith-Bao minimal integrity basis $\{ I_2, I_4, I_6, I_{10}\}$ just provides a support to Theorem \ref{prop:33}.

\section{The Four Invariants of the Basis are Algebraically Independent}

The next theorem claims that there is no syzygy relation among four invariants $J_2, J_4, J_6$ and $J_{10}$, where $\{J_2, J_4, J_6, J_{10} \}$ be an arbitrary minimal integrity basis of $\D$.

\begin{Theorem}
  Let $\{J_2, J_4, J_6, J_{10} \}$ be an arbitrary minimal integrity basis of a third order three-dimensional symmetric and traceless tensor $\D$.  Then there is no syzygy relation among four invariants $J_2, J_4, J_6$ and $J_{10}$.
\end{Theorem}
\begin{proof}   We first show that there is no syzygy relation among four invariants $I_2, I_4, I_6$ and $I_{10}$, where $\{I_2, I_4, I_6, I_{10} \}$ is the Smith-Bao minimal integrity basis of $\D$.

For a given third order three-dimensional symmetric and traceless tensor $\D$, we define
\begin{equation*}
  g(\x) := D_{ijk}x_ix_jx_k,
\end{equation*}
where $\x=(x_1,x_2,x_3)^\top$. Using seven independent elements of the tensor $\D$:
\begin{equation*}
  D_{111}, D_{112}, D_{113}, D_{122}, D_{123}, D_{222}, \text{ and } D_{223},
\end{equation*}
the homogeneous polynomial $g(\x)$ could be rewritten as
\begin{eqnarray*}
  g(\x) &=& D_{111} x_1^3+3 D_{112} x_1^2x_2 +3 D_{113} x_1^2x_3 +3 D_{122} x_1x_2^2 +6 D_{123} x_1 x_2 x_3+3
   \left(-D_{111}-D_{122}\right) x_1x_3^2 \\&&{}
   +D_{222} x_2^3+3 D_{223} x_2^2 x_3+3 \left(-D_{112}-D_{222}\right) x_2 x_3^2 +\left(-D_{113}-D_{223}\right) x_3^3.
\end{eqnarray*}
On the unit sphere $\{\x:x_ix_i=1\}$, the homogeneous polynomial $g(\x)$ has a maximizer. By rotating coordinates, we could place one maximizer at a point $(1,0,0)^\top$. Hence, the maximizer $\x=(1,0,0)$ satisfies the following system
\begin{equation*}\footnotesize
\left\{\begin{aligned}
  3 D_{111} x_1^2+6 D_{112} x_1 x_2+6 D_{113} x_1 x_3+3 D_{122} x_2^2+6 D_{123} x_2 x_3+3
   \left(-D_{111}-D_{122}\right) x_3^2 &= \lambda x_1, \\
  3 D_{112} x_1^2+6 D_{122} x_1 x_2+6 D_{123} x_1 x_3+3 D_{222} x_2^2+6 D_{223} x_2 x_3+3
   \left(-D_{112}-D_{222}\right) x_3^2 &= \lambda x_2, \\
  3 D_{113} x_1^2+6 D_{123} x_1 x_2+6 \left(-D_{111}-D_{122}\right) x_1 x_3+3 D_{223}
   x_2^2+6 \left(-D_{112}-D_{222}\right) x_2 x_3+3 \left(-D_{113}-D_{223}\right) x_3^2 &= \lambda x_3.
\end{aligned}\right.
\end{equation*}
Then, we get
\begin{equation*}
  D_{112}=D_{113}=0, \qquad D_{111}\ge 0,
\end{equation*}
and
\begin{eqnarray*}
  g(\x) &=& D_{111} x_1^3+3 D_{122} x_1x_2^2 +6 D_{123} x_1 x_2 x_3+3
   \left(-D_{111}-D_{122}\right) x_1x_3^2 \\&&{}
   +D_{222} x_2^3+3 D_{223} x_2^2 x_3-3D_{222}x_2 x_3^2 -D_{223}x_3^3.
\end{eqnarray*}
Since $g(0,-x_2,-x_3)=-g(0,x_2,x_3)$, $g(\x)$ must have a zero point in the circle $\{(0,x_2,x_3)^\top:x_2^2+x_3^2=1\}$. We may further rotate coordinates such that $g(0,1,0)=0$. Hence, we have
\begin{equation*}
  D_{222}=0.
\end{equation*}

In the new coordinate, the tensor $\D$ has four independent elements (with slightly abusing of notations)
\begin{equation*}
  D_{111}\ge0,D_{122},D_{123}, \text{ and }D_{223}.
\end{equation*}

Four isotropic invariants $I_2,I_4,I_6,I_{10}$ are indeed
\begin{eqnarray*}
  I_2 &=& 4 D_{111}^2+6 D_{122} D_{111}+6 D_{122}^2+6 D_{123}^2+4 D_{223}^2, \\
  I_4 &=& 2 (4 D_{111}^4+12 D_{122} D_{111}^3+(18 D_{122}^2+12 D_{123}^2+5
   D_{223}^2) D_{111}^2+12 D_{122} (D_{122}^2+D_{123}^2 \\&&{} +D_{223}^2)
   D_{111}+6 D_{122}^4+6 D_{123}^4+4 D_{223}^4+12 D_{123}^2 D_{223}^2+12 D_{122}^2
   (D_{123}^2+D_{223}^2)), \\
  I_6 &=& 4 (4 (D_{122}^2+D_{223}^2) D_{111}^4+8 D_{122}
   (D_{122}^2+D_{123}^2+3 D_{223}^2) D_{111}^3+(4 D_{122}^4+(8
   D_{123}^2 \\&&{} +37 D_{223}^2) D_{122}^2+4 D_{123}^4+D_{223}^4-3 D_{123}^2
   D_{223}^2) D_{111}^2+4 D_{122} (5 D_{122}^2-7 D_{123}^2) D_{223}^2
   D_{111} \\&&{} +4 (D_{122}^2+D_{123}^2){}^2 D_{223}^2),
\end{eqnarray*}
and
\begin{eqnarray*}
  I_{10} &=& -8 (8 (D_{122}^3-3 D_{122} D_{223}^2) D_{111}^7+4 (6
   D_{122}^4+(6 D_{123}^2-39 D_{223}^2) D_{122}^2-5 D_{223}^4 \\&&{} -6 D_{123}^2
   D_{223}^2) D_{111}^6+6 D_{122} (4 D_{122}^4+(8 D_{123}^2-73
   D_{223}^2) D_{122}^2+4 D_{123}^4-21 D_{223}^4 \\&&{} -8 D_{123}^2 D_{223}^2)
   D_{111}^5+(8 D_{122}^6+24 (D_{123}^2-26 D_{223}^2) D_{122}^4+3
   (8 D_{123}^4-28 D_{223}^2 D_{123}^2 \\&&{} -109 D_{223}^4) D_{122}^2+8
   D_{123}^6+D_{223}^6+72 D_{123}^2 D_{223}^4+84 D_{123}^4 D_{223}^2)
   D_{111}^4 \\&&{} -2 D_{122} D_{223}^2 (231 D_{122}^4+2 (69 D_{123}^2+101
   D_{223}^2) D_{122}^2-45 D_{123}^4-78 D_{123}^2 D_{223}^2) D_{111}^3 \\&&{} -6
   D_{223}^2 (28 D_{122}^6+(32 D_{123}^2+41 D_{223}^2) D_{122}^4+2
   (6 D_{123}^4-11 D_{123}^2 D_{223}^2) D_{122}^2+8 D_{123}^6 \\&&{} +9 D_{123}^4
   D_{223}^2) D_{111}^2-24 D_{122} D_{223}^2 (D_{122}^6-(D_{123}^2-3
   D_{223}^2) D_{122}^4-(5 D_{123}^4+14 D_{223}^2 D_{123}^2)
   D_{122}^2 \\&&{} -D_{123}^4 (3 D_{123}^2+D_{223}^2)) D_{111}+8
   (-D_{122}^6+15 D_{123}^2 D_{122}^4-15 D_{123}^4 D_{122}^2+D_{123}^6)
   D_{223}^4).
\end{eqnarray*}
We now consider the Jacobian of $\{I_2,I_4,I_6,I_{10}\}$ in variables $\{D_{111},D_{122},D_{123},D_{223}\}$:
\begin{equation*}
  \mathrm{Jac}=\left(\begin{array}{cccc}
    \frac{\partial I_2}{\partial D_{111}} & \frac{\partial I_2}{\partial D_{122}} & \frac{\partial I_2}{\partial D_{123}} & \frac{\partial I_2}{\partial D_{223}} \\
    \frac{\partial I_4}{\partial D_{111}} & \frac{\partial I_4}{\partial D_{122}} & \frac{\partial I_4}{\partial D_{123}} & \frac{\partial I_4}{\partial D_{223}} \\
    \frac{\partial I_6}{\partial D_{111}} & \frac{\partial I_6}{\partial D_{122}} & \frac{\partial I_6}{\partial D_{123}} & \frac{\partial I_6}{\partial D_{223}} \\
    \frac{\partial I_{10}}{\partial D_{111}} & \frac{\partial I_{10}}{\partial D_{122}} & \frac{\partial I_{10}}{\partial D_{123}} & \frac{\partial I_{10}}{\partial D_{223}}
  \end{array}\right).
\end{equation*}
By some calculations, the determinant of this Jacobian is
\begin{eqnarray*}
  \det(\mathrm{Jac}) &=& 27648 D_{123} (9 D_{111}^4+24 D_{122} D_{111}^3-24
   (D_{122}^2+D_{123}^2) D_{111}^2-32 D_{122} (3
   D_{122}^2+D_{123}^2) D_{111} \\&&{} +16 (-3 D_{122}^4-2 D_{123}^2
   D_{122}^2+D_{123}^4)) D_{223}^3 (16 (3
   D_{122}^2-D_{223}^2) D_{111}^8+32 (D_{122}^3 \\&&{} +3 D_{123}^2 D_{122})
   D_{111}^7-8 (18 D_{122}^4+3 (4 D_{123}^2+3 D_{223}^2) D_{122}^2-6
   D_{123}^4-5 D_{223}^4 \\&&{} -18 D_{123}^2 D_{223}^2) D_{111}^6-24 D_{122} (8
   D_{122}^4+(16 D_{123}^2-D_{223}^2) D_{122}^2+8 D_{123}^4+D_{223}^4 \\&&{} +3
   D_{123}^2 D_{223}^2) D_{111}^5-(64 D_{122}^6+48 (4 D_{123}^2-7
   D_{223}^2) D_{122}^4+3 (64 D_{123}^4+96 D_{223}^2 D_{123}^2 \\&&{} +7
   D_{223}^4) D_{122}^2+64 D_{123}^6+25 D_{223}^6+132 D_{123}^2 D_{223}^4+240
   D_{123}^4 D_{223}^2) D_{111}^4 \\&&{} +6 D_{122} D_{223}^2 (48 D_{122}^4+4
   (8 D_{123}^2-3 D_{223}^2) D_{122}^2-16 D_{123}^4+5 D_{223}^4-8 D_{123}^2
   D_{223}^2) D_{111}^3 \\&&{} +4 D_{223}^2 (16 D_{122}^6+6 (8 D_{123}^2-7
   D_{223}^2) D_{122}^4+(48 D_{123}^4+78 D_{223}^2 D_{123}^2+9
   D_{223}^4) D_{122}^2 \\&&{} +16 D_{123}^6+3 D_{123}^2 D_{223}^4+12 D_{123}^4
   D_{223}^2) D_{111}^2-8 D_{122} (D_{122}^2-3 D_{123}^2) D_{223}^4
   (12 D_{122}^2 \\&&{} -D_{223}^2) D_{111}-16 (D_{122}^3-3 D_{122}
   D_{123}^2){}^2 D_{223}^4),
\end{eqnarray*}
which is a polynomial in variables $\{D_{111},D_{122},D_{123},D_{223}\}$.
Clearly, the hypersurface $\det(\mathrm{Jac})= 0$ divides the space $\REAL^4$ of $(D_{111},D_{122},D_{123},D_{223})$ into several regions. We consider one of them.

Let $\Omega\subseteq\{(D_{111},D_{122},D_{123},D_{223})^\top : \det(\mathrm{Jac})\neq 0\}$ be a maximal connected open set, where ``maximal'' means that $\Omega$ can not be contained in another connected open set such that $\det(\mathrm{Jac})\neq 0$. As a polynomial in $D_{111},D_{122},D_{123}$ and $D_{223}$, $\det(\mathrm{Jac})\neq 0$ holds for all points in $\Omega$. Then, we process by contradiction. Suppose that there exists a syzygy relation among isotropic invariants $I_2,I_4,I_6,$ and $I_{10}$, which is denoted as a polynomial equation $$p(I_2,I_4,I_6,I_{10})=0.$$ Clearly, $p$ is also a polynomial in variables $D_{111},D_{122},D_{123}$ and $D_{223}$. By chain rule, we have
\begin{equation}\label{Chain-rule}
  \left(\begin{array}{c}
    \frac{\partial p}{\partial D_{111}} \\ \frac{\partial p}{\partial D_{122}} \\ \frac{\partial p}{\partial D_{123}} \\ \frac{\partial p}{\partial D_{223}}
  \end{array}\right) = \frac{\partial p}{\partial I_2}
  \left(\begin{array}{c}
    \frac{\partial I_2}{\partial D_{111}} \\ \frac{\partial I_2}{\partial D_{122}} \\ \frac{\partial I_2}{\partial D_{123}} \\ \frac{\partial I_2}{\partial D_{223}}
  \end{array}\right) + \frac{\partial p}{\partial I_4}
  \left(\begin{array}{c}
    \frac{\partial I_4}{\partial D_{111}} \\ \frac{\partial I_4}{\partial D_{122}} \\ \frac{\partial I_4}{\partial D_{123}} \\ \frac{\partial I_4}{\partial D_{223}}
  \end{array}\right) + \frac{\partial p}{\partial I_6}
  \left(\begin{array}{c}
    \frac{\partial I_6}{\partial D_{111}} \\ \frac{\partial I_6}{\partial D_{122}} \\ \frac{\partial I_6}{\partial D_{123}} \\ \frac{\partial I_6}{\partial D_{223}}
  \end{array}\right) + \frac{\partial p}{\partial I_{10}}
  \left(\begin{array}{c}
    \frac{\partial I_{10}}{\partial D_{111}} \\ \frac{\partial I_{10}}{\partial D_{122}} \\ \frac{\partial I_{10}}{\partial D_{123}} \\ \frac{\partial I_{10}}{\partial D_{223}}
  \end{array}\right)=0.
\end{equation}
Clearly, $\frac{\partial p}{\partial I_2}, \frac{\partial p}{\partial I_4}, \frac{\partial p}{\partial I_6},$ and $ \frac{\partial p}{\partial I_{10}}$ are polynomials in variables $D_{111},D_{122},D_{123}$ and $D_{223}$. Since $\det(\mathrm{Jac})\neq 0$ for all points in $\Omega$, we know that four one-way arrays in the middle of \eqref{Chain-rule} are linear independent. Hence, we have
\begin{equation*}
  \frac{\partial p}{\partial I_2}=\frac{\partial p}{\partial I_4}=\frac{\partial p}{\partial I_6}=\frac{\partial p}{\partial I_{10}}=0.
\end{equation*}
Therefore, the polynomial $p$ is a constant function in $\Omega$ whose value is zero.

By a similar discussion, we obtain that $p$ is a constant function in every region. Since $p$ is a polynomial, we get that $p$ must be a zero function. This contradicts the assumption that there exists a syzygy relation among isotropic invariants $I_2,I_4,I_6,$ and $I_{10}$.

We now show that there is no syzygy relation among four invariants $J_2, J_4, J_6$ and $J_{10}$, where $\{J_2, J_4, J_6, J_{10} \}$ is an arbitrary minimal integrity basis of $\D$.  Suppose that there exists a syzygy relation among isotropic invariants $J_2,J_4,J_6,$ and $J_{10}$, which is denoted as a polynomial equation $$q(J_2,J_4,J_6,J_{10})=0.$$  Since $\{I_2, I_4, I_6, I_{10} \}$ is an integrity basis of $\D$, we may represent $J_2, J_4, J_6$ and $J_{10}$ as polynomials of $I_2, I_4, I_6$ and $I_{10}$.   Note that in this way $J_2$ should be a polynomial of $I_2$, $J_4$ should be a polynomial of $I_2$ and $I_4$, etc.   Thus, we have polynomial function relations:
$$J_2 =J_2(I_2),$$
$$J_4 =J_4(I_2,I_4),$$
$$J_6 =J_6(I_2,I_4,I_6),$$
$$J_{10} =J_{10}(I_2,I_4,I_6,I_{10}).$$
Then we have a syzygy relation among isotropic invariants $I_2,I_4,I_6,$ and $I_{10}$ as follows:
$$q(J_2(I_2),J_4(I_2,I_4),J_6(I_2,I_4,I_6),J_{10}(I_2,I_4,I_6,I_{10}))=0.$$
This forms a contradiction.   Hence, there is no syzygy relation among four invariants $J_2, J_4, J_6$ and $J_{10}$.
\end{proof}

\textbf{Remark.} We note that the conclusion of algebraic independence among invariants forming an irreducible function basis of a tensor is not trivial. There exist syzygies in invariants forming an irreducible function basis of several tensors. For example, Chen, Liu, Qi, Zheng and Zou \supercite{CLQZZ-18} studied third order three-dimensional symmetric tensors and gave three syzygies among the eleven invariants of an irreducible function basis of isotropic invariants the symmetric tensors.


\section{Integrity and Function Bases of  Third Order Tensors}

In three-dimensional physical spaces, there are important third order tensors such as third order symmetric and traceless tensors, third order symmetric tensors, the Hall tensor, the piezoelectric tensor, etc.

In 1997, Smith and Bao\supercite{SB-97} presented a minimal integrity basis of four isotropic invariants for a third order three-dimensional symmetric and traceless tensor. Olive and Auffray\supercite{OA-14} constructed a minimal integrity basis with thirteen isotropic invariants for a third order symmetric tensor in 2014.  This year, Chen, Liu, Qi, Zheng and Zou\supercite{CLQZZ-18} showed that eleven isotropic invariants among the Olive-Auffray minimal integrity basis of a third order symmetric tensor form an irreducible function basis of that tensor. Also in this year, a ten invariant minimal integrity basis, which is also an irreducible function basis of the Hall tensor,  was presented by Liu, Ding, Qi and Zou\supercite{LDQZ-18}. For the piezoelectric tensor, in 2014, Olive\supercite{Ol-14} gave $495$ hemitropic invariants and claimed that these hemitropic invariants form an hemitropic integrity basis. Moreover, Olive\supercite{Ol-14} showed a set of $30,878$ isotropic invariants which form an integrity basis of isotropic invariants of the piezoelectric tensor.  Some further efforts are needed to find a function basis of the piezoelectric tensor with the cardinality smaller than the cardinality of the integrity basis given in \supercite{Ol-14}.


\end{document}